\documentclass[a4paper,11pt,reqno]{amsart}
\usepackage[a4paper]{geometry}
\usepackage{amssymb,amsmath}
\usepackage{mathtools}
\usepackage{enumerate}
\usepackage{bookmark}
\usepackage{hyperref}
\usepackage[initials,backrefs]{amsrefs}

\numberwithin{equation}{section}
\theoremstyle{plain}
\newtheorem{theorem}{Theorem}
\newtheorem{lemma}{Lemma}

\theoremstyle{definition}

\newtheorem*{assh*}{(H) Singular distributions}
\newcommand{\condH}{\mathbf{(H)}}
\newcommand{\prob}[1]{\DP\left\{#1\right\}}
\newcommand{\esm}[1]{\mathbb{E}\left[\,#1\,\right]}

\newcommand{\BC}{\mathbf{C}}
\newcommand{\BG}{\mathbf{G}}
\newcommand{\BH}{\mathbf{H}}

\newcommand{\BU}{\mathbf{U}}
\newcommand{\BV}{\mathbf{V}}

\newcommand{\BP}{\mathbf{P}}
\newcommand{\CJ}{\mathcal{J}}

\newcommand{\DP}{\mathbb{P}}
\newcommand{\DR}{\mathbb{R}}
\newcommand{\DZ}{\mathbb{Z}}
\newcommand{\BDelta}{\mathbf{\Delta}}
\newcommand{\BPsi}{\mathbf{\Psi}}
\newcommand{\Bx}{\mathbf{x}}
\newcommand{\By}{\mathbf{y}}

\newcommand{\FB}{\mathfrak{B}}
\newcommand{\ee}{\mathrm{e}}
\DeclareMathOperator{\dist}{dist}

\begin{document}
\title[Wegner bounds for multi-particle Bernoulli-Anderson models]{Wegner bounds for $N$-body interacting Bernoulli-Anderson models in one dimension}

\author[T.~Ekanga]{Tr\'esor EKANGA$^{\ast}$}

\address{$^{\ast}$%
Institut de Math\'ematiques de Jussieu, 
Universit\'e Paris Diderot,
Batiment Sophie Germain,
13 rue Albert Einstein,
75013 Paris,
France}
\email{tresor.ekanga@imj-prg.fr}
\date{\today}

\begin{abstract}
Under the weak interaction regime, we prove the one and the two volumes Wegner type bounds for one-dimensional multi-particle models on the lattice and for very singular probability distribution functions such as the Bernoulli's measures. The results imply the Anderson localization in both the spectral exponential and the strong dynamical localization for the one dimensional multi-particle Bernoulli-Anderson model with weak interaction.
\end{abstract}
\maketitle

\section{Introduction}

We are interested in this work to multi-particle Anderson models on the lattice space under the weak interaction regime and more important, with very singular probability distributions.

Different forms of the Wegner estimates for the single-particle as well as for the multi-particle theory have been obtained in earlier works, see for example \cites{CL90,CS08,HK13,K08,St01,W81} for the most famous results. In both of the papers, the common probability distribution functions of the i.i.d. random variables in the Anderson model was finally assumed to be at least H\"older continuous or even log-H\"older continuous in the road map to localization. In the above papers, the proofs were done with different strategies: some of them using the Feymann-Hellmann formula for operators \cite{K08} and others using the Stollmann's Lemma \cites{CS08,St01} or the scale-free unique continuation principle \cite{HK13}. We also refer to \cite{W81}, for the original Wegner bounds.

The Wegner estimates play an important role in the general strategy of Anderson localization with methods based on the fractional moment method of local resolvent as well as the one based on the multi-scale analysis in order to bound the probability of some resonant effects.

Localization for the Anderson models with Bernoulli distributions was obtained for one dimensional systems on the lattice by Carmona et al. \cite{CKM87} and in the continuum, by Damanik et al. \cite{DSS02}. Bourgain and Kenig \cite{BK05}, themselves prove localization at low energy for the multi-dimensional Bernoulli-Anderson model in the continuum using a strong form of the unique continuation principle in the continuum.

In some earlier works \cites{E11,E13,E16}, we prove localization using the Wegner type bounds established in the paper \cite{CS08} for the discrete case and in \cites{BCSS10,BCS11} for the continuum case.

We prove in the present paper , the stability of the Wegner bounds from the one dimensional single-particle system to weakly interacting multi-particle systems and for Bernoulli probability distribution functions of the i.i.d. stochastic random processes. The general strategy uses a perturbation argument based  on the resolvent identities for operators in Hilbert spaces.

Actually, what makes the machine work is that the one dimensional single-particle Wegner bounds are very strong (the decay rate is exponential). While, in the multi-scale analysis, only some polynomial decay bounds of the resonant effects are needed.

As a consequence of the Wegner bounds, we also obtain in the sequel, the Anderson localization for the one dimensional Bernoulli-Anderson model on the lattice. 

We stress that the multi-particle Wegner bounds of the present work do not need an additional assumption called separability of the bounded domains. The weak interacting regime itself is sufficient to prove our multi-particle Wegner type bounds without any condition on the domains. While in some numerous previous works \cites{CS08,HK13}, separability was crucial in order to prove the Wegner estimates and it results in some complications for the multi-scale analysis to prove the Anderson localization because in that case, one also has to analyze the cubes belonging to the bad region corresponding to the cubes which are not separable. What makes the machine works is that for a given cube , there is a finite number of cubes which are not separated  with the first one.

The main difference in our case is that, we prove the variable energy Wegner type bounds in intervals of small amplitude and such bounds are sufficient for the multi-scale analysis under the weak interaction regime see \cite{E13}.

Now, let us describe the structure of the paper. In the next Section, we set-up the model, the assumptions and formulate the main results. Section 3, is devoted to the road map to the proof of the main results. Finally, in Section 4, we prove the results on the Wegner estimates and discuss on the Anderson localization.

\section{The model and the main results}

\subsection{Our basic model and geometry}
We define the two following norms on $\ell^2(\DZ^D)$ for arbitrary $D\geq 1$: $|\Bx|=\max_{i=1,\ldots,D}|x_i|$ and $|\Bx|_1=|x_1|+ \cdots+|x_D|$. We consider a system of $N$-particles where $N\geq 2$ is finite and fixed. Let $d\geq 1$. We will prove the Wegner bounds for the random Hamiltonians $\BH^{(n)}_h(\omega)$ of the form

\begin{equation}\label{eq:def.H}
\BH^{(n)}_h(\omega)=-\BDelta + \sum_{j=1}^n V(x_j,\omega)+h\BU=-\BDelta + \BV(\Bx,\omega)+h\BU,
\end{equation}

acting on $\ell^2((\DZ^d)^n)\cong \ell^2(\DZ^{nd})$ with $h\in\DR$ and $\Bx\in(\DZ^d)^n$. Above, $\BDelta$ is the $nd$-dimensional lattice nearest neighbor Laplacian:

\begin{equation}
(\BDelta\BPsi)(\Bx)=\sum\limits_{\substack{\By\in\DZ^{nd}\\|\By-\Bx|_1=1}}(\BPsi(\By)-\BPsi(\Bx))=\sum\limits_{\By\in\DZ^{nd}\\|\By-\Bx|_1=1} \BPsi(\By)-2nd\BPsi(\Bx),
\end{equation}
for $\BPsi\in\ell^2(\DZ^{nd})$ and $\Bx\in \DZ^{nd}$. $V: \DZ^d\times \Omega\rightarrow  \DR$ is a random field relative to a probability space $(\Omega,\FB,\DP)$ and $\BU: (\DZ^d)^n\rightarrow \DR$ is the potential of the inter-particle interaction, $\BU$ is a bounded operator acting as multiplication operator on $\ell^2(\DZ^{nd})$ by the function $\BU(\Bx)$. $\BV$also acts on $\ell^2(\DZ^{nd})$ as a multiplication operator on $\ell^2(\DZ^{nd})$ by the function $\BV(\Bx,\omega)$.

For any $\Bx\in\DZ^{nd}$ and $L>0$, we denote by $\BC^{(n)}_L(\Bx)$, the $n$-particle cube in $\ell^2(\DZ^{nd})$, i.e., $\BC^{(n)}_L(\Bx):=\{\By\in\DZ^{nd}: |\Bx-\By|\leq L\}$. We also denote by $\sigma(\BH^{(n)}_h(\omega))$ , the spectrum of $\BH^{(n)}_h(\omega)$ and by $\BG^{(n)}_{\BC^{(n)}_L(\Bx),h}(E)$ the resolvent operator of $\BH^{(n)}_{\BC^{(n)},h}(\omega)$ for $E\notin \sigma(\BH^{(n)}_{\BC^{(n)}_L(\Bx),h}(\omega))$ and where the restriction is taken with Dirichlet boundary conditions.

\subsection{The assumption}
We denote by $\mu$, the common probability distribution measure of the i.i.d., random variables $\{V(x,\omega)\}_{x\in\DZ^d}$.

\begin{assh*}
We assume that the support of the measure $\mu$ is not concentrated in a single point and $\int|x|^{\eta}d\mu(x)<+\infty$ for some $\eta>0$.
\end{assh*}

\subsection{The results}

\begin{theorem}[One volume Wegner bound]\label{thm:Wegner.1.volume}

Let $d=1$ and $\BC^{(n)}_L(\Bx)$ be an $n$-particle cube in $\ell^2(\DZ^{nd})$. Assume that hypothesis $\condH$ holds true. Let $I\subset\DR$ be a compact interval. For any $0<\beta<1$, there exist $L_0=L_0(I,\beta)$ and $h^*=h^*(\|\BU\|,L_0)$ such that for all $h\in(-h^*,h^*)$,
\[
\prob{\dist(E,\sigma(\BH^{(n)}_{\BC^{(n)}_L(\Bx),h}(\omega)))\leq \ee^{-L^{\beta}}}\leq L^{-q},
\]
for all $E\in I$, $L\geq L_0$ and any $q>0$.
\end{theorem}

\begin{theorem}[Two volumes Wegner bounds]\label{thm:Wegner.2.volume}
Let $d=1$ and consider two $n$-particle cubes $\BC^{(n)}_L(\Bx)$ and $\BC^{(n)}_L(\By)$ in $\ell^2(\DZ^{nd})$. Assume that hypothesis $\condH$ holds true. Then for any $E_0\in\DR$ and any $0<\beta<1$, there exist $L_0=L_0(\beta)>0$, $\delta_0=\delta_0(\|\BU\|,L_0)>0$ such that if we put $I_0:=[E_0-\delta_0,E_0+\delta_0]$, then there exists $h^*=h^*(\|\BU\|,L_0)>0$ such that for all $h\in(-h^*,h^*)$:
\[
\prob{\exists E\in I_0: \max(\dist(E,\sigma(\BH^{(n)}_{\BC^{(n)}_L(\Bx),h}(\omega)));\dist(E,\sigma(\BH^{(n)}_{\BC^{(n)}_L(\By),h}(\omega))))\leq \ee^{-L^{\beta}}}\leq L^{-q},
\]
for all $L\geq L_0$ and any $q>0$.
\end{theorem}

\section{General strategy of the proofs}
We begin with the result on the Wegner estimates for single-particle models established in the paper by Carmona et al. \cite{CKM87}. The estimate is a very stronger result, namely the decay bound  is exponential. While for the Anderson localization via the multi-scale analysis some polynomial decay of the probability of the local resolvent is sufficient  to prove  localization. The strategy used in this Section in order to prove our main results is completely different to the proof of \cite{CKM87}. Our idea uses some perturbation argument on the local resolvent identities for operators in Hilbert spaces while the work \cite{CKM87}, is based on the study of the Lyapounov exponent and the theory of random transfer matrices for single-particle models in one dimension.

\subsection{The single-particle fixed energy Wegner bound}
The known result for single-particle models and for singular probability distributions including the Bernoulli's measures is the following:

\begin{theorem}\label{thm:1p.Wegner}
Let $I\subset \DR$ be a compact interval and $C^{(1)}_L(x)$ be a cube in $\DZ^d$. Assume that hypothesis $\condH$ holds true, then for any $0<\beta<1$ and $\sigma>0$, there exist $L_0=L_0(I,\beta,\sigma)>0$ and $\alpha=\alpha(I,\beta,\sigma)>0$ such that,
\[
\prob{\dist(E,\sigma(\BH^{(1)}_{C^{(1)}_L(x)}(\omega)))\leq \ee^{-\sigma L^{\beta}}}\leq \ee^{-\alpha L^{\beta}}
\]
for all $E\in I$ and all $L\geq L_0$.
\end{theorem}
\begin{proof}
We refer the reader to the paper by Carmona et al. \cite{CKM87}.
\end{proof}

\subsection{The non-interacting multi-particle Wegner bound}
We need to introduce first 
\[
 \{ (\lambda^{(i)}_{j_i},\psi^{(i)}_{j_i}): j_i=1,\ldots,|C^{(1)}_L(x_i)|\},
\]
the eigenvalues and the corresponding eigenfunctions of $H^{(1)}_{C^{(1)}_L(x_i)}(\omega)$, $i=1,\ldots,n$. Then, the eigenvalues $E_{j_1\ldots j_n}$ of the non-interacting multi-particle random Hamiltonian $\BH^{(n)}_{\BC^{(n)}_L(\Bx)}(\omega)$ are written as sums:

\[
E_{j_1\ldots j_n}=\sum_{i=1}^{n} \lambda_{j_i}^{(i)}=\lambda_{j_1}^{(1)}+\cdots+\lambda^{((n)}_{j_n},
\]
while the corresponding eigenfunctions $\BPsi_{j_1,\ldots,j_n}$ can be chosen as tensor products
\[
\BPsi_{j_1,\ldots,j_n}=\phi^{(1)}_{j_1}\otimes\cdots\otimes\psi^{(n)}_{j_n}.
\]
We also denote by $\lambda_{\neq i}:=\sum_{\ell=1, \ell\neq i}^n \lambda^{(\ell)}_{j_{\ell}}$. The eigenfunctions of finite volume Hamiltonians are assumed normalized.

The result of this subsection is 
\begin{theorem}\label{thm:np.0.Wegner}
Let $I\subset\DR$ be a compact interval and consider an $n$-particle cube $\BC^{(n)}_L(\Bx)\subset\DZ^{nd}$. Assume that hypothesis $\condH$ holds true, then for any $0<\beta<1$ and $\sigma>0$, there exist $L_0=L_0(I,\beta,\sigma)>0$ and $\alpha=\alpha(I,\beta,\sigma)>0$ such that
\[
\prob{\dist(E, \sigma(\BH^{(n)}_{\BC^{(n)}_L(\Bx)}(\omega)))\leq \ee^{-\sigma L^{\beta}}}\leq n\times(2L+1)^{nd}\times \ee^{-\alpha L^{\beta/p}},
\]
for all $E\in I$ and all $L\geq L_0$.
\end{theorem}

Before giving the proof, we have a Lemma.

\begin{lemma}\label{lem:np.0.Wegner}
Let $E\in\DR$ and $\BC^{(n)}_L(\Bx)\subset \DZ^{nd}$. Assume that for any $\sigma>0$, $0<\beta<1$ and $
i=1,\ldots,n$
\[
\dist(E-\lambda_{\neq i}, \sigma(H^{(1)}_{C^{(1)}_L(x_i)}(\omega)))>\ee^{-\sigma L^{\beta}},
\]
then, 
\[
\dist(E,\sigma(\BH^{(n)}_{\BC^{(n)}_L(\Bx)})) > \ee^{-2\sigma L^{\beta}}.
\]
\end{lemma}

\begin{proof}
In absence of the interaction, the multi-particle random Hamiltonian decomposes as 

\[
\BH^{(n)}_{\BC^{(n)}_L(\Bx),0}(\omega)=H^{(1,n)}_{C^{(1)}_L(x_1)}\otimes^{(n-1)}I+\cdots+I^{(n-1)}\otimes H^{(n,n)}_{C^{(1)}_L(x_n)}(\omega).
\]

For the $n$-particle resolvent operator $\BG^{(n)}_{\BC^{(n)}_L(\Bx)}(E)$, we also have the decomposition

\[
\BG^{(n)}_{\BC^{(n)}_L(\Bx)}(E)=\sum_{E_j\in\sigma(\BH^{(n)}_{\BC^{(n)}_L(\Bx)}} \BP_{\BPsi_{\neq j}}\otimes G^{(1)}_{C^{(1)}_L(x_j)}(E-\lambda_{\neq j}),
\]
where $\BP_{\BPsi_{\neq j}}$ denotes the projection onto the eigenfunction $\BPsi_{\neq j}=\bigotimes_{\ell=1,\ell\neq j}^{n}\varphi_{\ell}$ and $\{\lambda_i : i=1,\ldots,n\}$ are the eigenvalues of $H^{(1)}_{C^{(1)}_L(x_i)}(\omega), i=1,\ldots,n$ respectively and where $\lambda_{\neq j}=\sum_{i\neq j} \lambda_i$. Therefore, assuming all the eigenfunctions of finite volume Hamiltonians normalized, if each cube $C^{(1)}_L(x_i)$, $i=1,\ldots,n$ satisfies:

\[
\dist(E-E_{\neq i}; \sigma(H^{(1)}_{C^{(1)}_L(x_i)}))> \ee^{-\sigma L^{\beta}},
\]

then, we have the following upper bound for the multi-particle resolvent operator using the assertion of Theorem \ref{thm:1p.Wegner}, 

\[
\|\BG^{(n)}_{\BC^{(n)}_L(\Bx)}(E)\|\leq (2L+1)^{nd} \ee^{\sigma L^{\beta}},
\]
since the norm of the projection operator is bounded by $1$. Now, for sufficiently large $L>0$, the right hand side in the above equation can be bounded by $\ee^{2\sigma L^{\beta}}$ which proves the required bound.
\end{proof}

\begin{proof}[Proof of Theorem \ref{thm:np.0.Wegner}]
We have two cases:

\begin{enumerate}
\item[Case (a)]
For all $i=1,\ldots,n$, $x_i=x_1$, so that $\Bx=(x_1,\ldots,x_1)$. Using the decomposition of the $n$-particle Hamiltonian without interaction, we can identify

\begin{align*}
\BH^{(n)}_{\BC^{(n)}_L(\Bx),0}(\omega)&= H^{(1)}_{C^{(1)}_L(x_1)}(\omega)+\cdots+H^{(1)}_{C^{(1)}_L(x_1)}(\omega)\\
&=n\cdot H^{(1)}_{C^{(1)}_L(x_1)}(\omega).
\end{align*}
So that $\sigma(\BH^{(n)}_{\BC^{(n)}_L(\Bx)}(\omega))=n\cdot \sigma(H^{(1)}_{C^{(1)}_L(x_1)}(\omega))$. Therefore, 

\begin{align*}
\dist(E,\sigma(\BH^{(n)}_{\BC^{(n)}_L(\Bx)}))&= \dist(E; n\sigma(H^{(1)}_{C^{(1)}_L(x_1)}))\\
&= n\cdot\dist(\frac{E}{n}; \sigma(H^{(1)}_{C^{(1)}_L(x_1)}(\omega))).
\end{align*}

Hence, applying Theorem \ref{thm:1p.Wegner}, we get
\begin{align*}
\prob{\dist(E,\sigma(\BH^{(n)}_{\BC^{(n)}_L(\Bx),0}(\omega)))\leq \ee^{-\sigma L^{\beta}}}&\leq  \prob{n\cdot\dist(\frac{E}{n},\sigma(H^{(1)}_L(x_1)(\omega)))\leq \ee^{-\sigma L^{\beta}}}\\
&\leq \prob{\dist(\frac{E}{n},\sigma(H^{(1)}_{C^{(1)}_L(x_1)}(\omega)))\leq \frac{1}{n} \ee^{-\sigma L^{\beta}}}\\
&\leq \prob{\dist(\frac{E}{n},\sigma(H^{(1)}_{C^{(1)}_L(x_1)}(\omega)))\leq \ee^{-\sigma L^{\beta}}}\\
&\leq \ee^{-\alpha L^{\beta}}.
\end{align*}

\item[ Case (b)]

The single-particle cubes in the product $\BC^{(n)}_L(\Bx)=C^{(1)}_L(x_1)\times\cdots\times C^{(1)}_L(x_n)$ are not all the same. Thus, we begin with $C^{(1)}_L(x_1)$  and we assume that the total number of the single-particle projections that are the same with $C^{(1)}_L(x_1)$ including $C^{(1)}_L(x_1)$ itself is $n^{(1)}_0$ and the one of those different from $C^{(1)}_L(x_1)$ is $n^{(1)}_+$. Clearly, $n^{(1)}_0+n^{(1)}_+=n$. We denote by $C^{(1)}_L(y_1),\ldots,C^{(1)}_L(y_{n^{(1)}_+})$ the latter cubes which are different from $C^{(1)}_L(x_1)$. There exists  subset $\CJ_1\subset\{1,2,\ldots,(2L+1)^{d}\}$ such that
\[
\varPi_{\CJ_1} C^{(1)}_L(x_1),
\]
are the elements of $C^{(1)}_L(x_1)$ not belonging to $C^{(1)}_L(y_1)$. So that
\[
\varPi_{\CJ_1} C^{(1)}_L(x_1)\cap C^{(1)}_L(y_1)=\emptyset.
\]
We continue the procedure and find $\CJ_2\subset\{1,2,\ldots,(2L+1)^{d}\}$ such that 
\[
\varPi_{\CJ_2} C^{(1)}_L(x_1)\cap[ C^{(1)}_L(y_2)\setminus \varPi_{\CJ_1} C^{(1)}_L(x_1)]=\emptyset.
\]
At the end, we find $\CJ_{n_+}\subset \{1,2,\ldots,(2L+1)^d\}$ such that

\[
\varPi_{\CJ_{n_+}} C^{(1)}_L(x_1)\cap \left[C^{(1)}_L(y_{n_+})\setminus\left(\bigcup_{i=1}^{n_+-1} \varPi_{\CJ_i} C^{(1)}_L(x_1)\right)\right]=\emptyset.
\]

For all $\ell=1,\ldots, n^{(1)}_+$ we set,

 \[
    \Lambda^{(1,i)}_L:=\varPi_{\CJ_i} C^{(1)}_L(x_1);\qquad  D^{(1,\ell)}_L:=C^{(1)}_L(y_{\ell})\setminus\left(\bigcup_{i=1}^{\ell-1} \varPi_{\CJ_i} C^{(1)}_L(x_1)\right).
\]
 We therefore obtain
\[
\left[\bigcup_{i=1}^{n^{(1)}_+}\Lambda^{(1,i)}_L \right]\cap \left[ \bigcup_{\ell=1}^{n^{(1)}_+} D^{(1,\ell)}_L\right]=\emptyset
\]
and remark that each term in the above intersection  is non-empty. Now, we denote by $\FB_1$, the sigma algebra:

\[
\FB_1:=\Sigma\left(\left\{ V(x,\omega);  x\in\bigcup_{\ell=1}^{n^{(1)}_+} D^{(1,\ell)}_L\right\}\right).
\]
 Note that, at the $k^{th}$ step, we construct the domain $D^{(k)}$ in a similar way by finding the corresponding domain $\Lambda^{(k)}$  of the $n^{(k)}_+<n^{(k-1)}_+$ remaining cubes in the same way as above. Repeating this procedure, we reduce step by step  the number of the remaining cubes and arrive at the last step by building the domain $D^{(r)}_L$ with the required conditions and with the number $r\leq n$. In other words, the domains $\Lambda^{(k)}_L$ are non-empty and satisfy the conditions: for all $k\neq k'$, $\Lambda^{(k)}_L\cap\Lambda^{(k')}_L=\emptyset$. For all $1\leq k \leq r$, we also define the following sigma algebras:

\[
\FB_{k}:=\Sigma\left(\left\{V(x,\omega): x\in \bigcup_{\ell=1}^{n^{(k)}_+} D^{(k,\ell)}_L\right\}\right)
\]
Next, we define the sigma-algebra:
\[
\FB_{<\infty}:=\bigcup_{k=1}^r \FB_{k}.
\]
Now, observe that for any $i=1,\ldots,n$, we can find a sigma algebra $\FB_{j_i}$ for some $j_i\in\{1,\ldots,r\}$ such that the quantity $E_{\neq i}$ is $\FB_{j_i}$-measurable.  
By Lemma \ref{lem:np.0.Wegner}, we have that

\begin{gather*}
\prob{\dist(E,\sigma(\BH^{(n)}_{\BC^{(n)}_L(\Bx)}(\omega)))\leq \ee^{-\sigma L^{\beta}}}\\
\leq \prob{\exists i=1,\ldots,n; \exists E_{\neq i}:\dist(E-E_{\neq i},\sigma(H^{(1)}_{C^{(1)}_L(x_i)}(\omega)))\leq \ee^{-\frac{\sigma}{2} L^{\beta}}}.\\
\end{gather*}

Next, using Theorem \ref{thm:1p.Wegner} and the sigma-algebra $\FB_{<\infty}$, we finally obtain:

\begin{gather*}
\prob{\exists i=1,\ldots,n; \exists E_{\neq i}: \dist(E-E_{\neq i},\sigma(H^{(1)}_{C^{(1)}_L(x_i)}(\omega)))\leq \ee^{-\frac{\sigma}{2} L^{\beta}}}\\
 \leq \sum_{i=1}^n \sum_{E_i\in\sigma(\BH^{(n)}_{\BC^{(n)}_L(\Bx)}(\omega))}\esm{\prob{\dist(E-E_{\neq i},;\sigma(H^{(1)}_{C^{(1)}_L(x_i)}(\omega)))\leq \ee^{-\frac{\sigma}{2} L^{\beta}}\Big| \FB_{<\infty}}}\\
\leq n\times (2L+1)^{nd}\times \ee^{-\alpha L^{\beta}}.
\end{gather*}

This completes the proof of Theorem \ref{thm:np.0.Wegner}.
\end{enumerate}
\end{proof}

\subsection{ The interacting multi-particle Wegner bound}

The main result is 

\begin{theorem}\label{thm:np.Weak.Wegner}
Let $E\in\DR$ and an $n$-particle cube $\BC^{(n)}_L(\Bx)\subset\DZ^{nd}$. Assume that assumption $\condH$ holds true. Then for any $0<\beta<1$ and $\sigma>0$ there exist $L_0=L_0(\beta,\sigma)>0$ such that 
\begin{equation}
\prob{\dist(E,\sigma(\BH^{(n)}_{\BC^{(n)}_L(\Bx),h}(\omega)))\leq \ee^{-\sigma L^{\beta}}}\leq L^{-q},
\end{equation}
for all $L\geq L_0$ and any $q>0$.
\end{theorem}

\begin{proof}
We use the second resolvent identity. Let $E\in\DR$. We have 
\[
\BG^{(n)}_{\BC^{(n)}_L(\Bx),h}(E)=\BG^{(n)}_{\BC^{(n)}_L(\Bx),0}(E) + h\BU \BG^{(n)}_{\BC^{(n)}_L(\Bx),0}(E) \BG^{(n)}_{\BC^{(n)}_L(\Bx),h}(E).
\]
So that
\begin{equation}\label{eq:bound.resolvent}
\|\BG^{(n)}_{\BC^{(n)}_L(\Bx),h}(E)\|\leq \|\BG^{(n)}_{\BC^{(n)}_L(\Bx),0}(E)\|+ |h|\cdot\|\BU\|\|\BG^{(n)}_{\BC^{(n)}_L(\Bx),0}(E)\|\|\BG^{(n)}_{\BC^{(n)}_L(\Bx),h}(E)\|.
\end{equation}
Now, assume that 
\[
\dist(E,\sigma(\BH^{(n)}_{\BC^{(n)}_L(\Bx),0}(\omega))) > \ee^{-\sigma L_0^{\beta}},
\]
then 
\[
\|\BG^{(n)}_{\BC^{(n)}_L(\Bx),0}(E)\| \leq \ee^{\sigma L_0^{\beta}}.
\]
Therefore, it follows from \eqref{eq:bound.resolvent} that:
\[
\|\BG^{(n)}_{\BC^{(n)}_L(\Bx),h}(E)\| \leq \ee^{\sigma L_0^{\beta}}+ |h|\cdot\|\BU\|\cdot\ee^{\sigma L_0^{\beta}}\cdot \|\BG^{(n)}_{\BC^{(n)}_L(\Bx),h}(E)\|.
\]
Next, we choose the parameter $h>0$ in such a way that:
\begin{equation}\label{eq:cond.h}
|h|\|\BU\| \ee^{\sigma L_0^{\beta}}\leq \frac{1}{2}.
\end{equation}
Hence, setting
\[
h^*:=\frac{1}{2\|\BU\|\ee^{\sigma L_0^{\beta}}},
\]
the condition of \eqref{eq:cond.h}  is satisfied for all $h\in(-h^*,h^*)$. Thus, since $L\geq L_0$, we obtain that
\[
\|\BG^{(n)}_{\BC^{(n)}_L(\Bx),h}(E)\|\leq \ee^{\sigma L^{\beta}} + \frac{1}{2}\cdot \|\BG^{(n)}_{\BC^{(n)}_L(\Bx),h}(E)\|.
\]
Yielding,
\[
\|\BG^{(n)}_{\BC^{(n)}_L(\Bx),h}(E)\| \leq \ee^{\sigma L^{\beta}},
\]
for a new $\beta$ bigger than the previous one, provided that $L_0$ is large enough. Finally, using the notations and the result of Theorem \ref{thm:np.0.Wegner}, we bound the probability  for the interacting Hamiltonian as follows,

\begin{align*}
\prob{\dist(E,\sigma(\BH^{(n)}_{\BC^{(n)}_L(\Bx),h}(\omega))) \leq \ee^{-\sigma L^{\beta}}}&\leq \prob{\dist(E,\sigma(\BH^{(n)}_{\BC^{(n)}_L(\Bx),0}(\omega))) \leq \ee^{-\sigma L_0^{\beta}}}\\
&\leq \prob{\dist(E,\sigma(\BH^{(n)}_{\BC^{(n)}_L(\Bx),0}(\omega))) \leq \ee^{-\sigma L^{\beta/p}}}\\
&\leq n\times (2L+1)^{nd} \times \ee^{-\alpha L^{\beta}/p}< L^{-q},
\end{align*}
for all $L\geq L_0$ and any $q>0$, provided that $L_0>0$ is large enough.
\end{proof}

\subsection{The multi-particle variable energy Wegner bound}
In this Section, we prove the variable energy Wegner bound on intervals of small amplitude and the main tool is the resolvent identity. Although they are not easy to obtain, the variable energy Wegner type bounds are useful for the variable energy multi-scale analysis in order to prove the Anderson localization. The result is,

\begin{theorem}\label{thm:np.var.Wegner}
Let $E_0\in\DR$ and a cube $\BC^{(n)}_L(x)$ in $\DZ^{nd}$. Assume that hypothesis $\condH$ holds true, then for any $0<\beta<1$ and any $\sigma>0$ there exist $L_0=L_0(E_0,\beta,\sigma)>0$, $\alpha=\alpha(E_0,\beta,\sigma)>0$ and $\delta_{0}=\delta_0(L_0,\beta)>0$  such that:
\[
\prob{\exists E\in[E_0-\delta_0;E_0+\delta_0]: \dist(E,\sigma(\BH^{(n)}_{\BC^{(n)}_L(\Bx)}(\omega)))\leq \ee^{-\sigma L^{\beta}}} \leq L^{-q},
\]
for all $L\geq L_0$ and any $q>0$.
\end{theorem}

\begin{proof}
By the first resolvent equation, we have 

\begin{equation}\label{eq:R.2}
\BG^{(n)}_{\BC^{(n)}_L(\Bx)}(E)= \BG^{(n)}_{\BC^{(n)}_L(\Bx)}(E_0)+ (E-E_0) \BG^{(n)}_{\BC^{(n)}_L(\Bx)}(E) \BG^{(n)}_{\BC^{(n)}_L(\Bx)}(E_0).
\end{equation}
We choose the parameters $\beta$, $L_0$ and $\sigma$ as in Theorem \ref{thm:1p.Wegner}. Now, if $\dist(E_0,\sigma(\BH^{(n)}_{\BC^{(n)}_L(\Bx)}(\omega)))> \ee^{-ç\sigma L_0^{\beta}}$, and $|E-E_0|\leq \frac{1}{2}\ee^{-\sigma L_0^{\beta}}$, then
\[
\dist(E,\sigma(\BH^{(n)}_{\BC^{(n)}_L(\Bx)}(\omega))) > \frac{1}{2} \ee^{-\sigma L^{\beta}}.
\]
Indeed, it follows from \eqref{eq:R.2} that
\begin{align*}
\|\BG^{(n)}_{\BC^{(n)}_L(\Bx)}(E)\|&\leq \|\BG^{(n)}_{\BC^{(n)}_L(\Bx)}(E_0)\|+ |E-E_0|\cdot\|\BG^{(n)}_{\BC^{(n)}_L(\Bx)}(E)\|\cdot\|\BG^{(n)}_{\BC^{(n)}_L(\Bx)}(E_0)\|\\
&\leq \ee^{\sigma L_0^{\beta}} + \frac{1}{2}\cdot \ee^{-\sigma L_0^{\beta}}\cdot\ee^{\sigma L_0^{\beta}}\cdot \|\BG^{(n)}_{\BC^{(n)}_L(\Bx)}(E)\|\\
&\leq \ee^{\sigma L_0^{\beta}} + \frac{1}{2} \|\BG^{(n)}_{\BC^{(n)}_L(\Bx)}(E)\|.
\end{align*}
Yielding, 
\[
\|\BG^{(n)}_{\BC^{(n)}_L(\Bx)}(E)\|-\frac{1}{2}\| \BG^{(n)}_{\BC^{(n)}_L(\Bx)}(E)\|\leq \ee^{\sigma L^{\beta}},
\]
since, $L\geq L_0$. Thus,
\[
\|\BG^{(n)}_{\BC^{(n)}_L(\Bx)}(E)\|\leq 2 \ee^{\sigma L^{\beta}},
\]
or in other terms with a new $\beta$ bigger  than the first one:
\[
\dist(E,\sigma(\BH^{(n)}_{\BC^{(n)}_L(\Bx)}(\omega))) > \ee^{-\sigma L^{\beta}}.
\]
Hence, setting $\delta_0:=\frac{1}{2} \ee^{-\sigma L_0^{\beta}}$, we bound the probability in the statement of Theorem \ref{thm:np.var.Wegner} as follows,

\begin{gather*}
\prob{\text{$|E-E_0|\leq \delta_0$ and $\dist(E,\sigma(\BH^{(n)}_{\BC^{(n)}_L(\Bx)}(\omega))) \leq \ee^{-\sigma L^{\beta}}$}}\\
\leq \prob{\dist(E_0,\sigma(\BH^{(n)}_{\BC^{(n)}_L(\Bx)}(\omega))) \leq \ee^{-\sigma L_0^{\beta}}}.
\end{gather*}
Next, for $L\geq L_0$, we can find  an integer $p=p(L_0,L)\geq 1$ in such a way  that $L_0^p\geq L$. Finally, using Theorem \ref{thm:1p.Wegner}, we obtain
\begin{align*}
&\prob{\dist(E_0,\sigma(\BH^{(n)}_{\BC^{(n)}_L(\Bx)}(\omega))) \leq \ee^{-\sigma L_0^{\beta}}}\\
&\qquad \qquad \leq \prob{\dist(E_0,\sigma(\BH^{(n)}_{\BC^{(n)}_L(\Bx)}(\omega))) \leq \ee^{-\sigma L^{\beta/p}}}\\
&\qquad \qquad \leq L^{-q},
\end{align*}
for all $L\geq L_0$ and any $q>0$.
\end{proof}

\section{Conclusion: proof of the results}

\subsection{Proof of Theorem \ref{thm:Wegner.1.volume}}
Clearly, the assertion of Theorem \ref{thm:np.var.Wegner} implies Theorem \ref{thm:Wegner.1.volume} because the probability of the fixed energy one volume Wegner bound is bounded by the one of the variable energy given in the assertion of Theorem \ref{thm:np.var.Wegner} and this completes the proof of Theorem \ref{thm:Wegner.1.volume}.

\subsection{Proof of Theorem \ref{thm:Wegner.2.volume}}
The result of Theorem \ref{thm:Wegner.2.volume} follows from the assertion of Theorem \ref{thm:np.var.Wegner}. Indeed, the probability on the variable energy Wegner bound for two cubes given in Theorem \ref{thm:Wegner.2.volume} is bounded by that of one of the cubes at the same energy. This proves the result.

\subsection{Consequence}
The above one and two volumes Wegner type bounds are sufficient to derive the Anderson localization in both the spectral exponential and the strong dynamical localization under the weak interaction regime of the $N$-body interacting disordered system. So, the work makes an important contribution in the field of the mathematics of disordered many body quantum systems.


\begin{bibdiv}

\begin{biblist}

\bib{BCS11}{article}{
   author={Boutet de Monvel, A.},
   author={Chulaevsky, V.},
	author={Suhov, Y.},
   title={Dynamical localization for multi-particle models with an alloy-type external random potential},   
	journal={Nonlinearity},
	volume={24},
   date={2011},
   pages={1451--1472},
}
\bib{BCSS10}{article}{
   author={Boutet de Monvel, A.},
   author={Chulaebsky, V.},
	 author={Stollmann, P.},
	 author={Suhov, Y.},
	title={Wegner type bounds for a multi-particle continuous Anderson model with an alloy-type external  random potential},
	journal={J. Stat. Phys.},
	volume={138},
	date={2010},
	pages={553--566},
}

\bib{BK05}{article}{
   author={Bourgain, J.},
   author={Kenig, C.},
	title={On localization in the continuous Anderson-Bernoulli model  in higher dimension},
   journal={Invent. Math.},
   volume={161},
   date={2005},
   pages={389--426},
}

\bib{CKM87}{article}{
   author={Carmona, R.},
   author={Klein, A.},
   author={Martinelli, F.},
   title={Anderson localization for Bernoulli and other singular potentials},
   journal={Commun. Math. Phys.},
   volume={108},
   date={1987},
   pages={41--66},
}
\bib{CL90}{book}{
   author={Carmona, R.},
   author={Lacroix, J.}, 
   title={Spectral Theory of Random Schr\"{o}dinger Operators},
   volume={20},
   publisher={Birkh\"auser Boston Inc.},
   place={Boston, MA},
   date={1990},
}

\bib{CS08}{article}{
   author={ Chulaevsky, V.},
   author={Suhov, Y.},
   title={Wegner bounds for a two particle tight-binding model},
   journal={Commun. Math. Phys.},
   volume={283},
   date={2008},
   pages={479--489},
}
\bib{DSS02}{article}{
   author={Damanik, D.},
   author={SimS, R.},
   author={Stolz, G.},
   title={Localization for one-dimensional, continuum, Bernoulli-Anderson models},
   journal={Duke Math. Journal},
   volume={114},
   date={2002},
   pages={59--100},
}

\bib{E11}{article}{
   author={Ekanga, T.},
   title={On two-particle Anderson localization at low energies},
   journal={C. R. Acad. Sci. Paris, Ser. I},
   volume={349},
   date={2011},
   pages={167--170},
}
\bib{E13}{misc}{
   author={Ekanga, T.},
	title={Multi-particle localization for weakly interacting Anderson tight-binding models},
	status={arXiv:math-ph/1312.4180},
	date={2013},
}

\bib{E16}{misc}{
   author={Ekanga, T.},
   title={Anderson localization for weakly interacting multi-particle models in the continuum},
   status={arXiv:math-ph/},
   date={2016},
}
\bib{HK13}{misc}{
    author={Hislop, P.},
		author={Klopp, F.},
		title={Optimal Wegner estimate and the density of states for $N$-body interacting Schr\"odinger operator  with random potentials},
		status={arxiv:math-ph/1310.6959},
		 date={2013},
}
\bib{K08}{article}{
   author={Kirsch, W.},
   title={An Invitation to Random Schr\"{o}dinger Operators},
  journal={Panorama et Synth\`eses, 25, Soc. Math. France, Paris},
	 date={2008},
}
\bib{St01}{book}{
   author={Stollmann, P.},
   title={Caught by disorder},
   series={Progress in Mathematical Physics},
   volume={20},
   note={Bound states in random media},
   publisher={Birkh\"auser Boston Inc.},
   place={Boston, MA},
   date={2001},
}
\bib{W81}{article}{
    author={Wegner, F.},
		title={Bounds on the density of states in disordered systems},
		journal={Z. Phys. B},
		volume={44},
		date={1981},
		pages={9--15},
}
\end{biblist}
\end{bibdiv}
\end{document}